\documentclass[submission]{eptcs}
\providecommand{\event}{TARK 2021} 
\usepackage{breakurl}             
\usepackage{underscore}           

\usepackage{amsfonts}
\usepackage{mathrsfs}
\usepackage{amssymb}
\usepackage{amsthm}
\usepackage{latexsym}
\usepackage{graphicx}
\usepackage{enumerate}
\usepackage{mathtools}
\usepackage{color}
\usepackage{float}
\usepackage{array}
\usepackage{eucal}

\newcommand{\shortv}[1]{}

\usepackage{tabularx}

\mathtoolsset{centercolon}

\usepackage{amssymb,tikz}

\newcommand{\draft}[1]{{\color{red}{\textsc{[#1]}}}}
\newcommand{\defin}[1]{\textbf{#1}}

\newcommand{\lthen}{\rightarrow}
\newcommand{\liff}{\leftrightarrow}

\newcommand{\defeq}{\coloneqq}
\newcommand{\dimp}{\Leftrightarrow}
\newcommand{\true}{\mathit{true}}
\newcommand{\false}{\mathit{false}}
\newcommand{\citeyear}{\cite}

\newcommand{\val}[1]{[\![ #1 ]\!]}
\newcommand{\commentout}[1]{}

\newcommand{\mult}[1]{\{\!\!\!\{ #1 \}\!\!\!\}}

\renewcommand{\phi}{\varphi}

\newcommand{\dop}{\mathit{do}}

\renewcommand{\L}{\mathcal{L}}
\newcommand{\A}{\mathcal{A}}

\newcommand{\F}{\mathcal{F}}

\newtheorem{theorem}{Theorem}
\newtheorem{lemma}{Lemma}

\newtheorem{innercustomlem}{Lemma}
\newenvironment{customlem}[1]
  {\renewcommand\theinnercustomlem{#1}\innercustomlem}
  {\endinnercustomlem}

\title{Language-based Decisions}
\author{Adam Bjorndahl
\institute{Department of Philosophy\\
Carnegie Mellon University\\
Pittsburgh, USA}
\email{abjorn@cmu.edu}
\and
Joseph Y. Halpern
\institute{Department of Computer Science\\
Cornell Univeristy\\
Ithaca, USA}
\email{halpern@cs.cornell.edu}
}
\def\titlerunning{Language-based Decisions}
\def\authorrunning{A.~Bjorndahl \& J.Y. Halpern}

\begin{document}

\maketitle

\begin{abstract}
In Savage's classic decision-theoretic framework \cite{Savage},
actions are formally defined as functions from states to outcomes. But
where do the state space and outcome space come from? Expanding on
recent work by Blume, Easley, and Halpern \cite{BEH06}, we consider a
language-based framework in which actions are identified with
(conditional) descriptions in a simple underlying language, while
states and outcomes (along with probabilities and utilities) are
constructed as part of a representation theorem. Our work expands the
role of language from that in \cite{BEH06} by using it not only for
the \textit{conditions} that determine which actions are taken, but also the
\textit{effects}. More precisely, we take the set of actions to be
built from those of the form $\dop(\phi)$, for formulas $\phi$ in the
underlying language. This presents a problem: how do we
interpret the result of $\dop(\phi)$ when $\phi$ is underspecified
(i.e., compatible with multiple states)? We answer this using tools
familiar from the semantics of counterfactuals \cite{Stalnaker68}:
roughly speaking, $\dop(\phi)$ maps each state to the ``closest''
$\phi$-state. This notion of ``closest'' is also something we
construct as part of the representation theorem; in effect, then, we
prove that (under appropriate assumptions) the agent is acting
\textit{as if} each underspecified action is first made definite and
then evaluated (i.e., by maximizing expected utility). Of course,
actions in the real world are often not presented in a fully precise
manner, yet agents reason about and form preferences among them all
the same. Our work brings the abstract tools of decision theory into
closer contact with such real-world scenarios. 
\end{abstract}

\section{Motivation}

In Savage's classic decision-theoretic framework \cite{Savage}
\emph{actions} are formally defined as functions from \emph{states} to
\emph{outcomes}. States are conceptualized as encoding the possible
uncertainty the decision-maker may have about the world, while
outcomes correspond intuitively to the payoff-relevant ways things
might turn out. Thus, an action $\alpha$ can be viewed as a kind of
long list: for each way the world might be (i.e., each state $s$),
$\alpha$ specifies what will happen---namely, the outcome
$\alpha(s)$---in case action $\alpha$ is actually performed in state $s$.

One might ask: where do the state space and outcome space come from?
Is it reasonable to model an agent using a mathematical apparatus they
presumably have no access to? Questions like these tap into a long
tradition of challenging the idealizations involved in models like
Savage's (see, e.g.,
\cite{Ahn07,AE07,DLR01,Ghir01,GS04,Karni06,Kreps92,Lip99,Machina03,tverskykoehler}). One
response   might be that we are 
not trying to \textit{duplicate} the decision-making process going on
``in the agent's head'', but rather to \textit{represent} it,
mathematically---to show that under certain conditions it can be
tracked with a certain type of formalism (in this case, as a form of
expected utility maximization). 

Although this reply might assuage some worries about the use of
abstract mathematical frameworks for reasoning about decision making
in general, it remains problematic that actions---the objects over
which agents are supposed to ``reveal'' their preferences, through
concrete, binary choices---cannot themselves be described except by
reference to the background state and outcome spaces, which might not
be the states and outcomes that the agent is actually thinking of.
In such models, 
although outcomes are what agents are supposed to ultimately care
about, actions are the \textit{means} by which they bring outcomes
about. This makes an agent's preferences regarding actions arguably
the closest point of contact that these models have to the empirical,
observable reality of choosing between alternatives. Indeed, this
interpretation of actions is what underlies many of the intuitions
brought to bear to justify the various axioms of decision making that
Savage postulates and relies upon to prove his celebrated
representation theorem. 

The concern with where the states and outcomes are coming from
motivated Blume, Easley, and Halpern \citeyear{BEH06}
(henceforth BEH)
to consider a
model where acts and language are taken to be primary in a sense that
we explain shortly, while the state and outcome space are constructed
as part of the representation rather than specified exogenously. In
more detail, BEH assumed that acts were programs in a simple
programming language
formed by closing off a set of primitive programs
using {\bf if} \ldots {\bf then} \ldots {\bf else} \ldots,
so that if $a$ and $b$ are programs and $t$ is a
test (intuitively, a formula in a propositional language), then 
\textbf{if $t$ then $a$ else $b$} is a program.
Thus, rather than conditioning actions on events (i.e., subsets of
a state space), they are conditioned on \textit{descriptions} of
events, namely, tests.
This approach allows BEH to not only circumvent a fixed, exogenous
specification of the state space and outcome space (instead,
they are constructed as part of a representation theorem, and programs
are identified with maps from from these states to outcomes),
but also (as they illustrate with several examples) makes it possible to
capture a variety of \textit{framing} effects, which basically derive
from a mismatch between how the modeler conceives of the world and how
the agent does, as manifested in different ways that descriptions of
events might map onto actual events.

Our work is perhaps best understood as an extension of their work
in which the role of language is even more central. Specifically, while
BEH allowed arbitrary primitive programs, we take the primitive
programs to have the form $\dop(\phi)$, where $\phi$ is a formula.  The
$\dop(\phi)$ notation follows Pearl \citeyear{Pearl.Biometrika};
intuitively, $\dop(\phi)$ means that the agent somehow makes $\phi$
true. Note that this action is somewhat underspecified; it does not
say what else becomes true as a result of $\phi$ being true; for example,
if $\psi$ is independent of $\phi$, it does not tell us whether $\psi$
or $\lnot \psi$ is true. In our representation theorem, we assume that
the agent has a way of specifying the effects of $\dop(\phi)$.
In more detail, we take states in our state space to be characterized
by formulas in the language (this is
similar to
the canonical model used
in BEH's representation theorem), and take the outcome space to be the
same as the state space, so that a program maps states to states.
As part of the representation theorem, the agent must decide what
state $\dop(\phi)$ maps each state $\omega$ to. We follow standard
approaches to giving semantics to counterfactuals \cite{Stalnaker68}
by taking 
$\dop(\phi)$ to map $\omega$ to the
``closest'' state to $\omega$ (according to some measure of closeness)
where $\phi$ is true.  Of course, what counts as
``closest'' depends on the agent's subjective view of the world,
and is constructed from their preferences over acts.

This approach allows us to model choices in a way
that seems to us closer to how agents perceive and reason about the
options available to 
them. To illustrate, consider a policy-maker trying to decide whether
to raise the minimum wage to \$15 or to leave it as is.  In our
framework, this amounts to comparing the acts $\dop(MW = \$15)$ and
$\dop(\true)$ (where $\dop(true)$ amounts to doing nothing). Of course,
different agents may disagree about the side-effects of increasing the
minimum wage (businesses may close, there may be more automation so jobs
may be lost, and so on). This amounts to saying that different agents
will interpret $\dop(MW = \$15)$ differently as a function from states
to states, although all will agree that it will result in a state where
the minimum wage is \$15.%
\footnote{We remark that in this paper we consider only the
  single-agent case, but we find the multi-agent case, and
  specifically the effect of disagreements about what the closest
  state is, an exciting direction for future work.}
We can also express  contingent policies in our framework, for example,
raising the minimum wage if the economy is healthy.

By making both the acts and the test conditions formulas, we can 
capture framing and coarseness effects not only in the test
conditions, but also in the choices. For example, we might imagine
agents reacting differently to statements like
``we will require that every citizen is paid at least \$15 dollars for
each hour they work'' versus
``we will require every business owner to pay their employees at least
\$15 for each hour they work'', even if we can see that these are
equivalent statements. Our framework would allow this.  

The rest of this paper is organized as follows. We present our
approach as an extension of the work of BEH. This has the benefit of
allowing us to apply their representation theorem directly and focus our
efforts on the novel aspects of our extension.
We begin in Section \ref{sec:lam} by reviewing the relevant
definitions from BEH and augmenting them with the new ones we need to
capture language-based, underspecified effects of actions. Then in
Section \ref{sec:rep} we articulate the representation theorem we are
aiming at, introduce decision-theoretic axioms that allow us to
achieve it---including axioms from BEH (Section \ref{sec:canc}) as
well as several new axioms (Section \ref{sec:selax})---and finally
prove the theorem (Section \ref{sec:thm}). Section \ref{sec:fur}
concludes with a discussion of future work.
Appendix \ref{app:prf} collects proofs omitted from the main text.

\section{Language, Actions, and Models} \label{sec:lam}


Our first step is to import the relevant definitions from BEH so as to
present our extension of their work in context. In order to emphasize
the changes that we make and to streamline the presentation, we alter
some of their notation and terminology, and focus on the special case
of their system without randomization. 

Let $\Phi$ denote a finite set of \emph{primitive propositions}, and
$\L = \L(\Phi)$ the propositional language consisting of all Boolean
combinations of these primitives. Although of course it is possible
(and interesting) to consider other languages, in this work we 
focus on languages of this form as the \emph{underlying language of
action}---intuitively, the language in which both the conditions and
the results of actions are specified. 

A \defin{basic model (over $\L(\Phi)$)} is a tuple $M = (\Omega,
\val{\cdot}_{M})$ where $\Omega$ is a nonempty set of \emph{states}
and $\val{\cdot}_{M}: \Phi \to 2^{\Omega}$ is a \emph{valuation
function}. The valuation is recursively extended to all formulas in
$\L$ in the usual way. Intuitively,
$\val{\phi}_M$ is the set of states where $\phi$ is true. Using
$\val{\cdot}_M$ allows us to
interpret descriptions in the language $\L$ (what BEH call ``tests'')
as events: $\phi$ is interpreted as the subset $\val{\phi}_{M}
\subseteq \Omega$ of the state space $\Omega$.
We sometimes drop the subscript when the model is clear from context,
and write $\omega \models \phi$ for $\omega \in \val{\phi}$. We say
that $\phi$ is \emph{satisfiable in $M$} if $\val{\phi}_{M} \neq
\emptyset$ and that $\phi$ is \emph{valid in $M$} if $\val{\phi}_{M}
= \Omega$, and write $\models \phi$ to indicate that $\phi$ is valid
in all basic models.
Finally, we define the \emph{theory of $\omega$ (in $M$)} to be the set of all formulas true at $\omega$, denoted $Th(\omega) = \{\phi \: : \: \omega \models \phi\}$, and write $\omega \equiv \omega'$ iff $Th(\omega) = Th(\omega')$.


Up to now, everything we have defined has followed BEH
exactly---their ``primitive tests'' are our primitive propositions
$\Phi$; their ``tests'' are our formulas $\L(\Phi)$; their ``test
interpretations'' are our valuations $\val{\cdot}_{M}$. Next we
define our version of their ``primitive choices''. This is where
our development begins to diverge, since we take these to be
actions of the form $\dop(\phi)$; in other words, we
specify primitive choices using the same underlying language
$\L(\Phi)$ that corresponds to tests, rather than treating them as a
brand new set of primitives. 

Formally, given a finite set of formulas $F \subseteq \L$, the set of
\defin{actions (over $F$)}, denoted by $\A_{F}$, is defined
recursively as follows: for each $\phi \in F$, $\dop(\phi)$ is an
action (called a \emph{primitive action}), and for all $\psi \in \L$
and $\alpha, \beta \in \A_{\L}$, 
\textbf{if $\psi$ then $\alpha$ else $\beta$} is an action.
Following BEH, we take $F$ to be finite (who take the set of
primitive choices to be finite). It is also convenient because it
allows us to exclude logical inconsistencies from $F$, obviating the
need to interpret actions like $\dop(\false)$. For the propositional
languages under consideration in this paper, up to logical
equivalence, there are only finitely many formulas in any case.

Naturally, we also wish to \textit{interpret} our actions in a way
that respects their connection to the underlying language.
This is the topic we turn to next.

\subsection{Selection models}

In a given basic model $M$, we want $\dop(\phi)$ to correspond to a
function whose range is contained in $\val{\phi}_{M}$, the set of
$\phi$-states. Thus, we restrict our attention to basic
models in which each $\phi \in F$ is satisfiable---in this case we say
that $M$ is \defin{$F$-rich}. But this is not enough: as discussed,
$\dop(\phi)$ is underspecified; it does not in general determine a
unique function. In order to interpret such actions and compare them
to others, we must in some sense ``fill in'' the missing details.
We formalize this with the concept of a \defin{selection model (for
$F$)}, which is a basic model $M = (\Omega, \val{\cdot}_{M})$ together
with a \emph{selection function (for $M$)} $c: \Omega \times F \to
\Omega$ satisfying $c(\omega,\phi) \in \val{\phi}_M$.

Selection functions were introduced by Stalnaker \citeyear{Stalnaker68}
as a mechanism to interpret counterfactual conditionals.  
Following this tradition,
we think of $c(\omega,\phi)$ as representing the ``closest'' state to
$\omega$ where $\phi$ is true.
There are many other properties one might insist $c$ have, aside from
$c(\omega, \phi) \in \val{\phi}$ (which is called
\defin{success}). For example, one may require that if $\omega \in
\val{\phi}$, then $c(\omega, \phi) = \omega$ (i.e., if $\phi$ is true
in $\omega$, then the closest state to $\omega$ where $\phi$ is true
is $\omega$ itself); this property is called \defin{centering}. 

In this paper we will also consider a relatively strong condition on
$c$, namely, that it is derived from a parametrized family of
\emph{well-orders}\footnote{A binary relation $\leq$ on a set is
  called a \emph{linear order} if it is complete, transitive, and
  antisymmetric (i.e., $x \leq y$ and $y \leq x$ implies $x=y$). A
  \emph{well-order} is a linear order in which every nonempty subset
  has a least element.} on the state space, one for each state: $\leq
\; \defeq \{\leq_{\omega} \: : \: \omega \in \Omega\}$. Intuitively,
$\omega_{1} \leq_{\omega} \omega_{2}$ says ``$\omega_{1}$ is at least
as close to $\omega$ as $\omega_{2}$ is''. We say that a selection
function $c$ is \defin{induced by $\leq$} if $c(\omega,\phi)$ always
outputs the $\leq_{\omega}$-minimal element of $\val{\phi}$. We call
$\leq$ \defin{centered} if, for each $\omega \in \Omega$, the
$\leq_{\omega}$-minimal element of $\Omega$ is $\omega$ (in which case
it is also easy to see that the induced selection function satisfies
centering). Finally, we say that $\leq$ is \defin{language-based} if
the relations $\leqq_{\omega}$ on the quotient $\Omega/\!\!\equiv$
given by 
$$[\omega_{1}] \leqq_{\omega} [\omega_{2}] \textrm{ iff } \omega_{1} \leq_{\omega} \omega_{2}$$
are well-defined well-orders, and moreover, whenever $\omega \equiv
\omega'$, we have $\leqq_{\omega} \, = \, \leqq_{\omega'}$. Note that
in this case $\omega \equiv \omega'$ implies $c(\omega,\phi) \equiv
c(\omega',\phi)$.\footnote{Here's why: since $c(\omega, \phi)$ is the
  $\leq_{\omega}$-minimal element of $\val{\phi}$, it must also be
  that $[c(\omega,\phi)]$ is the $\leqq_{\omega}$-minimal element of
  $\{[\omega''] \: : \: \omega'' \models \phi\}$. Similarly,
  $[c(\omega',\phi)]$ is the $\leqq_{\omega'}$-minimal element of
  $\{[\omega''] \: : \: \omega'' \models \phi\}$. Since
  $\leqq_{\omega} \, = \, \leqq_{\omega'}$, these must coincide, so we
  have $[c(\omega, \phi)] = [c(\omega', \phi)]$.}
Intuitively,
if $\leq$ is language-based then
what counts as the closest
state essentially depends only on the formulas that are true at a
state. We cannot have two states $\omega_1$ and $\omega_2$ that agree
on all formulas (so that $\omega_1 \equiv \omega_2$) and a third state
$\omega_3$ that does not agree with $\omega_1$ and $\omega_2$ on all
formulas such that $\omega_3$ is between $\omega_1$ and $\omega_2$ in
terms of distance from some state $\omega$ (i.e., we cannot have
$\omega_1 \le_\omega \omega_3 \le_\omega \omega_2$).


The purpose of the selection function in our models is to take an underspecified
transition from states to states and ``resolve the
ambiguity''. Specifically, given a transition that starts in state
$\omega$ and ends up in a $\phi$-state, the selection function $c$ can then
by applied to specify the exact $\phi$-state, namely
$c(\omega,\phi)$, where it actually ends up.
In this way, given a basic, $F$-rich model $M$, each action of the
form $\dop(\phi)$ can be interpreted in any selection model $(M,c)$
based on $M$ as a function $\val{\dop(\phi)}_{M,c}: \Omega \to \Omega$
defined by: 
$$\val{\dop(\phi)}_{M,c}(\omega) = c(\omega, \phi).$$
Of course, we can extend this interpretation to all actions in
$\A_{F}$ in the obvious way (and exactly as BEH do): 
$$
\val{\textbf{if $\psi$ then $\alpha$ else $\beta$}}_{M,c}(\omega) = \begin{cases}
\val{\alpha}_{M,c}(\omega) & \textrm{if $\omega \in \val{\psi}$}\\
\val{\beta}_{M,c}(\omega) & \textrm{if $\omega \notin \val{\psi}$.}
\end{cases}
$$


\section{Representation} \label{sec:rep}

We begin as usual with a binary relation $\succeq$ on $\A_{F}$, where
$\alpha \succeq \beta$ says that $\alpha$ is ``at least as good as''
$\beta$.
Following standard conventions, we
define $\alpha \succ \beta$ as an abbreviation for
$\alpha \succeq \beta$ and 
$\beta \not\succeq
\alpha$, and $\alpha \sim \beta$ for $\alpha \succeq \beta$ and $\beta
\succeq \alpha$, representing ``strict preference'' and
``indifference'', respectively.
We also assume that $\succeq$ is \emph{complete}, that is, all
elements are comparable, so that for all acts $\alpha$ and $\beta$,
either $\alpha \succeq \beta$ or $\beta \succeq\alpha$. 
Although BEH consider incomplete relations, we
focus here on the simpler case of complete relations in order to
streamline the presentation and highlight the novel components of our
model.

A \defin{language-based SEU (Subjective Expected Utility)
representation} for a relation $\succeq$ on $\A_{F}$ is a finite
selection model $(M,c)$ together with a probability measure $\pi$ on
$\Omega$ and a \emph{utility function} $u: \Omega \to \mathbb{R}$ such
that, for all $\alpha, \beta \in \A_{F}$, 
\begin{equation} \label{eqn:rep}
\alpha \succeq \beta \dimp \sum_{\omega \in \Omega} \pi(\omega) \cdot u(\val{\alpha}_{M,c}(\omega)) \geq \sum_{\omega \in \Omega} \pi(\omega) \cdot u(\val{\beta}_{M,c}(\omega)).
\end{equation}

We note the key differences between the representation theorem BEH
establish and what we are aiming at. First, their result produces a
separate outcome space and state space, whereas for us,
these spaces coincide. More importantly, their result treats
``primitive choices'' (namely, our actions $\dop(\phi)$, for $\phi \in
F$) as true primitives in the sense that each is assigned to an
\textit{arbitrary} function from states to outcomes. By contrast, we
want to respect the structure of an action like
$\dop(\phi)$---specifically, its connection to the formula $\phi$---by
requiring that $\dop(\phi)$ correspond to a map from $\Omega$ to
$\Omega$ such that $\omega \mapsto c(\omega, \phi)$ for a suitable
selection function $c$. One of the novel aspects of our proof consists in
showing how to determine the selection function from preferences on acts.

Since our framework can be viewed a specialization of the BEH
framework (with our actions having additional, language-based
structure as described), rather than proving our representation
theorem from scratch, we can reuse much of their construction. Thus,
we will present the same axioms (adapted to our notation) that BEH
present, and subseqently augment them with new principles that allow
us to construct the selection function.

\subsection{Cancellation} \label{sec:canc}

BEH's main axiom is a \textit{cancellation law}. Explaining this
requires a few preliminary definitions, beginning with
the notion of a \emph{multiset}, which
can be thought of as a set that allows for
multiple instances of each of its elements; two multisets are equal
just in case they contain the same elements \textit{with the same
  multiplicities}. For example, the multiset 
$\mult{a,a,a,b,b}$ is different from the multiset
$\mult{a,b,b,b,b}$: both multisets have five elements, but the
mulitiplicity of $a$ and $b$ differ.

Given any subset $X \subseteq \Phi$, let
$\phi_{X} = \bigwedge_{p \in X} p \land \bigwedge_{q \notin X} \lnot q$.
Intuitively, $\phi_{X}$ is a ``complete description'' of the truth
values of all primitive propositions in the language $\L(\Phi)$,
namely the description that says for each primitive proposition $p$
that it is true iff it belongs to $X$. An \defin{atom} is any formula
of the form $\phi_{X}$. Since $\L(\Phi)$ is a
propositional language and we use classical semantics for
propositional logic, for all formulas $\phi \in \L(\Phi)$ and atoms $\phi_{X}$, 
the truth of $\phi$ is determined by $\phi_{X}$: either $\models
\phi_{X} \lthen \phi$, or $\models \phi_{X} \lthen \lnot \phi$. It is
therefore not surprising that every action in $\alpha \in \A_{F}$
can be identified with a function $f_{\alpha}: 2^{\Phi} \to F$,
defined recursively as follows: 
\begin{eqnarray*}
f_{\dop(\phi)}(X) & = & \phi\\
f_{\textbf{if $\psi$ then $\alpha$ else $\beta$}}(X) & = & \begin{cases}
f_{\alpha}(X) & \textrm{if $\models \phi_{X} \lthen \psi$}\\
f_{\beta}(X) & \textrm{if $\models \phi_{X} \lthen \lnot \psi$.}
\end{cases}
\end{eqnarray*}
BEH define atoms in the same way and use them to define functions from
atoms to primitive choices just as we did above (replace $\dop(\phi)$
by an arbitrary primitive choice).%
\footnote{Techncially, we are not mapping atoms to primitive acts, but since
there is an obvious bijection $X \mapsto \phi_X$ beween sets of
primitive proposition and atoms, and an obvious bijection
$\phi \mapsto \dop(\phi)$ 
between elements of $F$ and primitive acts, we really can be thought of as
doing just that.}

Now we can state the central cancellation law that enables us to apply the BEH representation theorem:
\begin{description}
\item[(Canc)]
Let $\alpha_{1}, \ldots, \alpha_{n}, \beta_{1}, \ldots, \beta_{n} \in \A_{F}$, and suppose that for each $X \subseteq \Phi$ we have
$\mult{f_{\alpha_{1}}(X), \ldots, f_{\alpha_{n}}(X)} = \mult{f_{\beta_{1}}(X), \ldots, f_{\beta_{n}}(X)}$.
Then, if for all $i < n$ we have $\alpha_{i} \succeq \beta_{i}$, it follows that $\beta_{n} \succeq \alpha_{n}$.
\end{description}
Intuitively, this says that if we get the same collection of outcomes
with $\alpha_1, \ldots, \alpha_n$ as with $\beta_1,\ldots, \beta_n$
(taking multiplicity into account) in each state, then we should view
the collection $\mult{\alpha_1, \ldots, \alpha_n}$ and
$\mult{\beta_1, \ldots, \beta_n}$ as equally good.  Thus, if
$\alpha_i$ is at least as good as $\beta_i$ for $i=1, \ldots, n-1$,
then, to balance things out, $\beta_n$ should be at least as good as $\alpha_n$.

As pointed out by BEH, Cancellation is a surprisingly powerful axiom.
In particular, BEH show that 
we can use {\bf (Canc)} to derive many simpler (and
more classical) principles of choice: that $\succeq$ is reflexive and
transitive, that \emph{independence} holds,\footnote{That is, for all
  $\alpha, \beta, \gamma, \gamma' \in \A_{F}$ and all $\phi \in F$, 
$$(\textbf{if $\phi$ then $\alpha$ else $\gamma$} \succeq \textbf{if
    $\phi$ then $\beta$ else $\gamma$}) \; \dimp \; (\textbf{if $\phi$
    then $\alpha$ else $\gamma'$} \succeq \textbf{if $\phi$ then
    $\beta$ else $\gamma'$}).$$} and that if $\alpha$ and $\beta$ are
\emph{equivalent} in the sense that $f_{\alpha} = f_{\beta}$, then
$\alpha \sim \beta$.
(However, it should be noted that Cancellation seems stronger than the
conjunection of these axioms.)

\subsection{Selection axioms} \label{sec:selax}

To present the new axioms that will allow us to construct an
appropriate selection function as part of the representation theorem,
it will be helpful to introduce some new notation. To begin, we write
\textbf{if $\phi$ then $\alpha$} as a shorthand for \textbf{if $\phi$
then $\alpha$ else $\dop(\true)$}. Intuitively, the action
$\dop(\true)$ corresponds to doing ``nothing'', since $\true$ is
true no matter what, so we might think of ``otherwise nothing'' as
being the default in case no explicit {\bf else...} clause is
given. Of course, for this to make sense we must have $\true \in F$;
we make this assumption henceforth. 

Next we define an abbreviation for \emph{conditional preference}, familiar from Savage's classical development \cite{Savage}: write $\alpha \succeq_{\phi} \beta$ as an abbreviation for $(\textbf{if $\phi$ then $\alpha$}) \succeq (\textbf{if $\phi$ then $\beta$})$.\footnote{As BEH show, the cancellation law implies independence, so in fact we have $\alpha \succeq_{\phi} \beta$ iff for \textit{all} $\gamma$, $\textbf{if $\phi$ then $\alpha$ else $\gamma$} \succeq \textbf{if $\phi$ then $\beta$ else $\gamma$}$.}
When $\phi = \phi_{X}$, we write $\alpha \succeq_{X} \beta$ for $\alpha \succeq_{\phi_{X}} \beta$, and we extend this notation to strict conditional preference and conditional indifference in the obvious way.

Our first axiom is related to the centering constraint for selection functions (i.e., that if $\phi$ is true at a state, then that state automatically counts the ``closest'' $\phi$-state):

\begin{description}
\item[(Cent)]
If $\models \psi \lthen \phi$, then $(\textbf{if $\psi$ then $\dop(\phi)$}) \sim \dop(\true)$.
\end{description}

To build intuition it's helpful to consider the special case where
$\psi = \phi$, in which case {\bf (Cent)} just says that doing $\phi$
precisely when $\phi$ is already the case (and otherwise doing
nothing) is the same as doing nothing. Here of course by ``the same''
what is really meant is that the agent is indifferent between those
two acts. Since we are trying to bootstrap properties of a
selection function from the agent's preferences, all our principles
will ultimately need to bottom out in statements about what the agent
does or does not have a preference between. The general statement of
{\bf(Cent)} simply expands this reasoning to cases where the
condition $\psi$ entails the result of the action, $\phi$, and so
again in this case $\dop(\phi)$ happens only in cases where $\phi$ is
already true. 

\begin{lemma} \label{lem:cent}
If $(M,c)$ is a selection model, $c$ satisfies
centering, and $\models \psi \lthen \phi$, then
$$\val{\textbf{if $\psi$ then $\dop(\phi)$}}_{M,c} = id_{\Omega} = \val{\dop(\true)}_{M,c}.$$
\end{lemma}

Our second axiom is meant to capture the idea that \textit{sufficiently specific conditions} resolve any ambiguity (expressible in the underlying language) about the effect of an action:
\begin{description}
\item[(SSC)]
If $\models \phi \liff (\phi_{1} \lor \cdots \lor \phi_{n})$, then $\forall X \subseteq \Phi$, $\exists i \in \{1, \ldots, n\}$ such that for all $\psi$ satisfying $\models \phi_{i} \lthen \psi$ and $\models \psi \lthen \phi$, we have $\dop(\psi) \sim_{X} \dop(\phi_{i})$.
\end{description}

This requires some unpacking. As above, it is illuminating to begin by considering the special case where $\psi = \phi$. Then $\models \psi \lthen \phi$ holds trivially and $\models \phi_{i} \lthen \psi$ is true by assumption, so we can read {\bf (SSC)} intuitively as follows: If $\phi$ is ambiguous between a variety of (potentially) more precise statements (namely, $\phi_{1}, \ldots, \phi_{n}$), then for any sufficiently specific condition (i.e., any atom $\phi_{X}$), there is at least one precisification $\phi_{i}$ of $\phi$ such that, conditional on $\phi_{X}$, doing $\phi$ is equivalent to doing $\phi_{i}$ (from the agent's perspective).

This, as well as the more general statement of {\bf (SSC)}, follows from the assumption that the selection function $c$ is induced by a language-based family of well-orders.
\begin{lemma} \label{lem:lang}
If $(M,c)$ is a selection model where $c$ is induced by the
well-orders $\leq \; = \{\leq_{\omega} \: : \: \omega \in \Omega\}$,
$\leq$ is language-based, $\models \phi \liff (\phi_{1}
\lor \cdots \lor \phi_{n})$, and $X \subseteq \Phi$, then $\exists i \in
\{1, \ldots, n\}$ such that for all $\psi$ satisfying $\models
\phi_{i} \lthen \psi$ and $\models \psi \lthen \phi$ and all $\omega
\in \val{\phi_{X}}$, we have
$\val{\dop(\psi)}_{M,c}(\omega) = \val{\dop(\phi_{i})}_{M,c}(\omega)$.
\end{lemma}

%
%
%
\commentout{
\begin{lemma}
{\bf (SSC)} implies that whenever $\models \phi \liff \phi'$, we also have $\dop(\phi) \sim \dop(\phi')$.
\end{lemma}

\begin{proof}
In the statement of {\bf (SSC)}, take $n=1$ and $\phi_{1} =
\phi'$. This yields $\dop(\phi) \sim_{X} \dop(\phi')$ for all $X
\subseteq \Phi$; the result now follows from
{\bf (Canc)}. See the full paper for details.
\end{proof}
}

The next idea is crucial to the ultimate construction of our selection function. For each atom $\phi_{W}$, we will define a total preorder\footnote{A \emph{total preorder} is a complete and transitive relation
(so, unlike a linear order, it need not be antisymmetric).}
$\sqsubseteq_{W}$ on the set of atoms that will in turn be extended to a linear order and used to specify the selection function. Formally, we define:
$$\phi_{X} \sqsubseteq_{W} \phi_{Y} \textrm{ iff } \dop(\phi_{X} \lor \phi_{Y}) \sim_{W} \dop(\phi_{X}).$$
Loosely speaking, $\phi_{X} \sqsubseteq_{W} \phi_{Y}$ says that in $\phi_{W}$-states, the ambiguity inherent in doing $\phi_{X} \lor \phi_{Y}$ is resolved in the agent's mind in favour of doing $\phi_{X}$; this is why the agent is indifferent (conditional on $\phi_{W}$) between doing $\phi_{X} \lor \phi_{Y}$ and just doing $\phi_{X}$. In this sense we think of $\phi_{X}$ as being at least as ``close'' to $\phi_{W}$ as $\phi_{Y}$ is.

Note that the definition above requires $F$ to contain all atoms as
well as all pairwise disjunctions of atoms. This richness in $F$ is
what allows us to use the agent's preferences on actions to define an
appropriate preorder. We make this assumption henceforth. It is an
interesting question to what extent the ensuing construction can be
carried out without this assumption; we return to this point in
Section \ref{sec:fur}. 

Now we can state our third axiom, which simply says that this notion of closeness is transitive:
\begin{description}
\item[(Trans)]
For all $W,X,Y,Z \subseteq \Phi$, if $\phi_{X} \sqsubseteq_{W} \phi_Y$ and $\phi_{Y} \sqsubseteq_{W} \phi_{Z}$, then $\phi_{X} \sqsubseteq_{W} \phi_{Z}$.
\end{description}


\begin{lemma} \label{lem:com}
{\bf (SSC)} implies that each $\sqsubseteq_{W}$ is complete.
\end{lemma}

\begin{lemma} \label{lem:wll}
If {\bf (SSC)} and {\bf (Trans)} hold, then each $\sqsubseteq_{W}$ is
a total preorder and can be extended to a well-order $\leq_{W}$ on the
set of atoms; if, in addition, {\bf (Cent)} holds, then each $\leq_{W}$
can be defined so that $\phi_{W}$ is the $\leq_{W}$-minimal element. 
\end{lemma}

Given a family of well-orders $\{\leq_{W} \: : \: W \subseteq \Phi\}$ as defined in Lemma \ref{lem:wll}, let $min_{\leq}(W, \phi)$ denote the unique $X \subseteq \Phi$ such that $\phi_{X}$ is $\leq_{W}$-minimal in $\{\phi_{Y} \: : \: \models \phi_{Y} \lthen \phi\}$. So $\phi_{X}$ is the ``closest'' atom compatible with $\phi$ to $\phi_{W}$; intuitively, then, doing $\phi$ in a $\phi_{W}$ situation should essentially amount to doing $\phi_{X}$. This is precisely what the next lemma asserts.

\begin{lemma} \label{lem:asif}
If {\bf (SSC)} and {\bf (Trans)} hold, then $\dop(\phi) \sim_{W} \dop(\phi_{min_{\leq}(W,\phi)})$.
\end{lemma}

\subsection{The representation theorem} \label{sec:thm}

\begin{theorem}
If $\succeq$ is a complete binary relation on $\A_{F}$ satisfying {\bf (Canc)}, {\bf (Cent)}, {\bf (SSC)}, and {\bf (Trans)}, then there is a language-based SEU representation for $\succeq$.
\end{theorem}

\begin{proof}
We begin by following the proof in \cite[Theorem 2]{BEH06} to obtain a
state-dependent representation with state space $2^{\Phi}$ and outcome
space $F$.\footnote{``State-dependent'' here means that the utility
  function constructed will depend not only on outcomes but on states
  as well.} More precisely, we consider the set of functions $\F =
\{f_{\alpha} \: : \: \alpha \in \A_{F}\}$ defined in Section
\ref{sec:canc}, which can be viewed as Savage acts in the classical
sense \cite{Savage}. The relation $\succeq$ on $\A_{F}$ induces a
relation $\succeq^{*}$ on $\F$ defined as follows: 
$$f_{\alpha} \succeq^{*} f_{\beta} \dimp \alpha \succeq \beta.$$
As discussed, {\bf (Canc)} implies that $\alpha \sim \alpha'$ whenever $f_{\alpha} = f_{\alpha'}$, so $\succeq^{*}$ is well-defined; moreover, as BEH show, {\bf (Canc)} is strong enough to yield the desired state-dependent representation result for $\succeq^{*}$, namely, that there exists a function $u^{*}: 2^{\Phi} \times F \to \mathbb{R}$ such that, for all $f,g \in \F$,
$$f \succeq^{*} g \dimp \sum_{X \in 2^{\Phi}} u^{*}(X, f(X)) \geq \sum_{X \in 2^{\Phi}} u^{*}(X, g(X)).$$

Up to now we have mirrored the proof given by BEH exactly, which has
given us a utility function $u^{*}$ but also an outcome space that we don't
want.
Moreover, the utility function is \emph{state-dependent}; it takes as
arguments both a state and an outcome. We want a utility function
that depends only on states (which for us are the same as outcomes).  
Thus, our task now is to transform this result into a selection model 
that we can use to give a language-based SEU representation of
$\succeq$ (including a utility function defined only on states).  

Set $\Omega = 2^{\Phi} \times 2^{\Phi}$; so our state space is
isomorphic to \textit{pairs} of atoms. This is a technical maneuver
that allows us to ``factor out'' probabilities from the
state-dependent utility function $u^{*}$ we already have. Loosely
speaking, given $(X,Y) \in \Omega$, the first component $X$ represents
how things are, while the second component $Y$ represents how things
\textit{were}. This intuition should become clearer as we continue. 

We define a basic model $M = (\Omega, \val{\cdot}_{M})$ by specifying the valuation on $\Omega$ as follows:
$$\val{p}_{M} = \{(X,Y) \in \Omega \: : \: \models \phi_{X} \lthen p\}.$$
In other words, $p$ is true at $(X,Y)$ just in case $\phi_{X}$ entails
$p$. Note that the valuation only depends on the first component $X$ of
the state $(X,Y)$. 

Next we specify a parametrized family of well-orders on $\Omega$ that
we can use to induce a selection function. First define
$$(X,X') \sqsubseteq_{W,W'} (Y,Y') \textrm{ iff } \phi_{X} \leq_{W} \phi_{Y}.$$
Again, we are ignoring the second component. This is clearly a
well-order when restricted to the first component of the state space,
but not in general, since by definition we have $(X,X')
\sqsubseteq_{W,W'} (Y,Y')$ and $(Y,Y') \sqsubseteq_{W,W'} (X,X')$
whenever $X=Y$. However, as usual, we can extend these relations to
well-orders $\leq_{W,W'}$ on all of $\Omega$ simply by choosing a
linear order for each set of the form $\Omega_{X} \defeq \{(X,Y) \: :
\: Y \in 2^{\Phi}\}$, and in so doing we can insist that for each
fixed $X$, the state $(X,W)$ is $\leq_{W,W'}$-minimal on the set
$\Omega_{X}$. 

This is the first time we have paid attention to the second component of the state. Roughly speaking, we are ensuring that the order $\leq_{W,W'}$ ``remembers'' the set $W$. More perspicuously, it is easy to see that if $c$ is the selection function induced by the family $\{\leq_{(W,W')} \: : \: (W,W') \in \Omega\}$, then for each $(W,W') \in \Omega$ and all $\phi \in \L$, we have
\begin{equation} \label{eqn:move}
\val{\dop(\phi)}_{M,c}(W,W') = c((W,W'),\phi) = (min_{\leq}(W,\phi), W).
\end{equation}
That is, the closest $\phi$-state to $(W,W')$ encodes both the closest atom compatible with $\phi$ to $\phi_{W}$ (in the first component) \textit{and} the state $W$ that we started from (in the second component).

Now we can define our utility function and probability measure. Let $\pi$ be any probability measure on $\Omega$ satisfying $\pi(\Omega_{X}) > 0$ for all $X$. Next, define $u: \Omega \to \mathbb{R}$ by
$$u(X,W) = \frac{u^{*}(W,\phi)}{\pi(\Omega_{W})}, \textrm{ for some $\phi$ such that $min_{\leq}(W,\phi) = X$}.$$
Of course, we need to check that $u$ is well-defined, and we do so in Lemma \ref{lem:wdf}. But first some intuition is in order. Thinking back to the state-dependent utility function $u^{*}$, a reasonable first gloss of the meaning of $u^{*}(W,\phi)$ might be ``the utility of doing $\phi$ in $W$''.\footnote{Though this isn't quite right---it's more like the product of that utility with the probability of $W$, which is why we have to factor that probability out in defining our utility function.} The point is that $u^{*}$ is specifying the utility value not of an action in itself or the ``result'' of an action, but rather the result of an action \textit{if you started in a certain state}. This is all very informal, but the idea is just to provide some intuition for why, in defining our utility function $u$ from $u^{*}$, we need to appeal to a rich enough notion of state that can ``remember'' what the ``previous'' state was---intuitively, the state we were at before the action was performed.

\begin{lemma} \label{lem:wdf}
The function $u$ is well-defined.
\end{lemma}

The last thing we need to show is that the selection model $(M,c)$ we have built, along with $\pi$ and $u$, gives us an expected utility representation of $\succeq$. So let $\alpha, \beta \in \A_{F}$ and suppose that $\alpha \succeq \beta$. By definition this is equivalent to $f_{\alpha} \succeq^{*} f_{\beta}$, which by the state-dependent representation result is in turn equivalent to
\begin{equation} \label{eqn:sdeu}
\sum_{W \in 2^{\Phi}} u^{*}(W, f_{\alpha}(W)) \geq \sum_{W \in 2^{\Phi}} u^{*}(W, f_{\beta}(W)).
\end{equation}
Now observe that, for each $W \in 2^{\Phi}$,
\begin{eqnarray*}
u^{*}(W, f_{\alpha}(W)) & = & \pi(\Omega_{W}) \cdot u(min_{\leq}(W,f_{\alpha}(W)),W) \qquad \qquad \textrm{(by definition of $u$)}\\
& = & \pi(\Omega_{W}) \cdot u(\val{\dop(f_{\alpha}(W))}_{M,c}(W,W')) \; \; \qquad \textrm{(from (\ref{eqn:move}))}\\
& = & \pi(\Omega_{W}) \cdot u(\val{\alpha}_{M,c}(W,W')) \; \qquad \qquad \qquad \textrm{(by definition of $f_{\alpha}$ and $(M,c)$).}
\end{eqnarray*}
Note that in the above $W'$ can be \textit{any} element of $2^{\Phi}$, since it's not taken into account in determining the result of an action. That means we can rewrite the above as
$$u^{*}(W, f_{\alpha}(W)) = \sum_{W' \in 2^{\Phi}} \pi(W,W') \cdot u(\val{\alpha}_{M,c}(W,W')).$$
Of course, an analogous equation holds for $u^{*}(W, f_{\beta}(W))$. Thus, (\ref{eqn:sdeu}) is equivalent to:
$$\sum_{W \in 2^{\Phi}} \sum_{W' \in 2^{\Phi}} \pi(W,W') \cdot u(\val{\alpha}_{M,c}(W,W')) \geq \sum_{W \in 2^{\Phi}} \sum_{W' \in 2^{\Phi}} \pi(W,W') \cdot u(\val{\beta}_{M,c}(W,W')),$$
which
is exactly the right-hand side of (\ref{eqn:rep}), completing the proof.
\end{proof}

\section{Discussion} \label{sec:fur}

We have considered a framework in which both the conditions for and
the results of an action are given by simple descriptions in a fixed
language. These descriptions may not be maximally specific, so the
results of actions can be underspecified and therefore ``open to
interpretation''. We have shown that, in this context, agents whose
preferences satisfy certain constraints can be represented as if they
are expected utility maximizers who interpret each underspecified
action using a selection function identical to that employed in
standard semantics for counterfactual conditionals. 

The representation theorem presented in this extended abstract might
be viewed as a sort of ``proof of concept'', namely, that such
representation results are possible and even natural. This opens the
door for a variety of related results connecting different assumptions
about the selection function to different constraints on the agent's
preferences.
As we mentioned above, there are a number of standard assumptions
along these lines in the literature on counterfactuals.

The underlying language we chose to work with can also be
altered. Perhaps most obviously, we might consider allowing
countably-many primitive propositions. In this case, we cannot
straightforwardly use atoms as the basis for the state space in the
representation theorem, and in general we might need to relax the
notion of a ``complete description'' to something like a
``sufficiently detailed description''. Going in the other direction,
we might also considering dropping some of the richness
constraints we imposed. For instance, we assumed that $F$ contains all
atoms (and all pairwise disjunctions of atoms). Can this assumption be
relaxed? 

In our framework, because we use the same descriptions for both states
and outcomes, we found it convenient to identify the two. This in turn
makes it
straightforward to extend to a richer language of acts, where we allow
\emph{sequential actions}, implemented
directly by function composition.
That is, we can allow actions of the form $\dop(\phi); \dop(\psi)$
(``first do $\phi$, then do $\psi$''), or more generally, $\alpha; \beta$.
Thus, the (underspecified!) results
of the first action are directly relevant to the conditions under
which the second action is executed, which may allow for entirely new
and intriguing ways of encoding modeling features via constraints on
preferences. 

Finally, generalizing this framework to multiple agents is of
interest. Indeed, the original motivation for this work is doubly
relevant in multi-agent settings: two different decision-makers might
conceive of the same action in different ways, by associating it with
different functions. For example, we should be able to model two
agents who agree about their values and have the same beliefs about
the likelihoods of uncertain events, but still have different
preferences over actions---intuitively, because they interpret the
``default'' way of implementing actions differently (in other words,
they have the same utility function and probability measure, but
different selection functions). 

In short, this area is ripe for further exploration, with many
theoretical and practical applications.

\appendix

\section{Proofs} \label{app:prf}

\begin{customlem}{\ref{lem:cent}}
If $(M,c)$ is a selection model, $c$ satisfies
centering, and $\models \psi \lthen \phi$, then
$$\val{\textbf{if $\psi$ then $\dop(\phi)$}}_{M,c} = id_{\Omega} = \val{\dop(\true)}_{M,c}.$$
\end{customlem}

\begin{proof}
By definition, we have
\begin{eqnarray*}
\val{\textbf{if $\psi$ then $\dop(\phi)$}}_{M,c}(\omega) & = & \begin{cases}
\val{\dop(\phi)}_{M,c}(\omega) & \textrm{if $\omega \in \val{\psi}$}\\
\val{\dop(\true)}_{M,c}(\omega) & \textrm{if $\omega \notin \val{\psi}$.}
\end{cases}\\
& = & \begin{cases}
c(\omega, \phi) & \textrm{if $\omega \in \val{\psi}$}\\
c(\omega, \true) & \textrm{if $\omega \notin \val{\psi}$.}
\end{cases}
\end{eqnarray*}
But since $\val{\psi} \subseteq \val{\phi}$ by assumption, in either case, centering applies and guarantees that
$$\val{\textbf{if $\psi$ then $\dop(\phi)$}}_{M,c}(\omega) = \omega. \qedhere$$
\end{proof}

\begin{customlem}{\ref{lem:lang}}
If $(M,c)$ is a selection model where $c$ is induced by the
well-orders $\leq \; = \{\leq_{\omega} \: : \: \omega \in \Omega\}$,
$\leq$ is language-based, $\models \phi \liff (\phi_{1}
\lor \cdots \lor \phi_{n})$, and $X \subseteq \Phi$, then $\exists i \in
\{1, \ldots, n\}$ such that for all $\psi$ satisfying $\models
\phi_{i} \lthen \psi$ and $\models \psi \lthen \phi$ and all $\omega
\in \val{\phi_{X}}$, we have 
$\val{\dop(\psi)}_{M,c}(\omega) = \val{\dop(\phi_{i})}_{M,c}(\omega)$.
\end{customlem}

\begin{proof}
Let $\omega \in \val{\phi_{X}}$ and choose $i$ such that
$c(\omega,\phi) \in \val{\phi_{i}}$. This is possible since we know
$c(\omega, \phi) \in \val{\phi}$ and, by assumption,
$\val{\phi} = \val{\phi_{1}} \cup \ldots \cup \val{\phi_{n}}$.
Since $c(\omega, \phi)$ is the $\leq_{\omega}$-minimal
element of $\val{\phi}$, it follows that for any set $T$ with
$c(\omega, \phi) \in T \subseteq \val{\phi}$, $c(\omega, \phi)$ is
also the $\leq_{\omega}$-minimal element of $T$. In particular, since
$c(\omega, \phi) \in \val{\phi_{i}} \subseteq \val{\psi} \subseteq
\val{\phi}$, this implies that $c(\omega, \phi)$ is the
$\leq_{\omega}$-minimal element of both $\val{\phi_{i}}$ and
$\val{\psi}$. Thus, by definition, $c(\omega, \phi_{i}) = c(\omega,
\psi)$, so 
$$\val{\dop(\psi)}_{M,c}(\omega) = c(\omega, \psi) = c(\omega, \phi_{i}) = \val{\dop(\phi_{i})}_{M,c}(\omega).$$
Since $\omega \models \phi_{X}$ and this completely determines the theory of $\omega$, we know that for any other $\omega' \in \val{\phi_{X}}$, $\omega' \equiv \omega$, so $c(\omega', \phi) \equiv c(\omega, \phi)$. This guarantees that $c(\omega', \phi) \in \val{\phi_{i}}$; in other words, the same choice of $i$ works for all states in $\val{\phi_{X}}$, which completes the proof.
\end{proof}

\begin{customlem}{\ref{lem:com}}
{\bf (SSC)} implies that each $\sqsubseteq_{W}$ is complete.
\end{customlem}

\begin{proof}
Fix any two atoms $\phi_{X}$ and $\phi_{Y}$. We apply {\bf (SSC)} in the case where $\phi = \phi_{X} \lor \phi_{Y}$, $\phi_{1} = \phi_{X}$, $\phi_{2} = \phi_{Y}$, and $\psi = \phi$. Then we know that given any $W \subseteq \Phi$, either $\dop(\phi) \sim_{W} \dop(\phi_{1})$ or $\dop(\phi) \sim_{W} \dop(\phi_{2})$, that is, either $\dop(\phi_{X} \lor \phi_{Y}) \sim_{W} \dop(\phi_{X})$ or $\dop(\phi_{X} \lor \phi_{Y}) \sim_{W} \dop(\phi_{Y})$, which established completeness.
\end{proof}

\begin{customlem}{\ref{lem:wll}}
If {\bf (SSC)} and {\bf (Trans)} hold, then each $\sqsubseteq_{W}$ is
a total preorder and can be extended to a well-order $\leq_{W}$ on the
set of atoms; if, in addition, {\bf (Cent)} holds, then each $\leq_{W}$
can be defined so that $\phi_{W}$ is the $\leq_{W}$-minimal element. 
\end{customlem}

\begin{proof}
The fact that $\sqsubseteq_{W}$ is a total preorder follows
immediately from {\bf (Trans)} and Lemma \ref{lem:com}.
Moreover, it is easy
to see that any total preorder on a finite set can be extended to a
well-order (by choosing an arbitrary linear order for each subset of
$\sqsubseteq_{W}$-equivalent atoms). To see that this can be done in
such a way that $\phi_{W}$ is the $\leq_{W}$-minimal element, it
suffices to show that for every $X \subseteq \Phi$, we have $\phi_{W}
\sqsubseteq_{W} \phi_{X}$, or in other words, $\dop(\phi_{W} \lor
\phi_{X}) \sim_{W} \dop(\phi_{W})$. The result now follows from two
applications of {\bf (Cent)}. First we apply it in the case where
$\psi = \phi_{W}$ and $\phi = \phi_{W} \lor \phi_{X}$ to obtain
$(\textbf{if $\phi_{W}$ then $\dop(\phi_{W} \lor \phi_{X})$}) \sim
\dop(\true)$; then we apply it in the case where $\psi = \phi =
\phi_{W}$ to obtain $(\textbf{if $\phi_{W}$ then $\dop(\phi_{W})$})
\sim \dop(\true)$. Transitivity of $\sim$ therefore yields
$$(\textbf{if $\phi_{W}$ then $\dop(\phi_{W} \lor \phi_{X})$}) \sim (\textbf{if $\phi_{W}$ then $\dop(\phi_{W})$}),$$
which by definition is equivalent to $\dop(\phi_{W} \lor \phi_{X}) \sim_{W} \dop(\phi_{W})$.
\end{proof}

\begin{customlem}{\ref{lem:asif}}
If {\bf (SSC)} and {\bf (Trans)} hold, then $\dop(\phi) \sim_{W} \dop(\phi_{min_{\leq}(W,\phi)})$.
\end{customlem}

\begin{proof}
Let $X = min_{\leq}(W,\phi)$, and let $\phi_{X_{1}}, \ldots, \phi_{X_{n}}$ enumerate all the atoms compatible with $\phi$. Then by definition we know that $X = X_{j}$ for some $j$. We also clearly have $\models \phi \liff (\phi_{X_{1}} \lor \cdots \lor \phi_{X_{n}})$, so we can apply {\bf (SSC)} (taking $\psi = \phi$) to find an $i$ such that $\dop(\phi) \sim_{W} \dop(\phi_{X_{i}})$.

By definition of $X$, we know that $\phi_{X} \leq_{W} \phi_{X_{i}}$, which means $\dop(\phi_{X} \lor \phi_{X_{i}}) \sim_{W} \dop(\phi_{X})$. On the other hand, since $\models \phi_{X_{i}} \lthen (\phi_{X} \lor \phi_{X_{i}})$ and $\models (\phi_{X} \lor \phi_{X_{i}}) \lthen \phi$, {\bf (SSC)} also tells us (taking $\psi = \phi_{X} \lor \phi_{X_{i}}$ this time) that $\dop(\phi_{X} \lor \phi_{X_{i}}) \sim_{W} \dop(\phi_{X_{i}})$. By transitivity of $\sim_{W}$ we therefore have $\dop(\phi_{X_{i}}) \sim_{W} \dop(\phi_{X})$, and therefore $\dop(\phi) \sim_{W} \dop(\phi_{X})$, as desired.
\end{proof}

\begin{customlem}{\ref{lem:wdf}}
The function $u$ is well-defined.
\end{customlem}

\begin{proof}
What we need to show that is that if $min_{\leq}(W,\phi) = X$ and also $min_{\leq}(W,\phi') = X$, then $u^{*}(W,\phi) = u^{*}(W,\phi')$. By Lemma \ref{lem:asif}, we know that $\dop(\phi) \sim_{W} \dop(\phi_{X})$, and also that $\dop(\phi') \sim_{W} \dop(\phi_{X})$. Focusing on the first of these two indifferences to begin with, by definition we have
$$\textbf{if $\phi_{W}$ then $\dop(\phi)$} \sim \textbf{if $\phi_{W}$ then $\dop(\phi_{X})$}.$$
Setting $\alpha = \textbf{if $\phi_{W}$ then $\dop(\phi)$}$ and $\beta = \textbf{if $\phi_{W}$ then $\dop(\phi_{X})$}$, it follows that $f_{\alpha} \sim^{*} f_{\beta}$ (by definition of $\succeq^{*}$). Thus, from the state-dependent representation result, we can deduce that
$$\sum_{Z \in 2^{\Phi}} u^{*}(Z, f_{\alpha}(Z)) = \sum_{Z \in 2^{\Phi}} u^{*}(Z, f_{\beta}(Z)).$$
But it's easy to see that whenever $Z \neq W$, $f_{\alpha}(Z) =
f_{\beta}(Z)$, so we can cancel all those terms in the equality above
to arrive at $u^{*}(W, f_{\alpha}(W)) = u^{*}(W, f_{\beta}(W))$. This
yields $u^{*}(W, \phi) = u^{*}(W, \phi_{X})$, since clearly
$f_{\alpha}(W) = \phi$ and $f_{\beta}(W) = \phi_{X}$. Analogous
reasoning starting from the fact that $\dop(\phi') \sim_{W}
\dop(\phi_{X})$ leads us to $u^{*}(W, \phi') = u^{*}(W,
\phi_{X})$. Putting these together gives $u^{*}(W, \phi) = u^{*}(W,
\phi')$, as desired. 
\end{proof}

\bibliographystyle{eptcs}
\bibliography{z,joe}
\end{document}